\documentclass[11pt]{llncs} 

\usepackage{fullpage}

\usepackage{amssymb}

\newcommand{\PR}[1]{\mathbf{Pr}\left[#1\right]} 
\newcommand{\E}[1]{\mathbf{E}\left[#1\right]} 
\newcommand{\eps}{\epsilon} 
\newcommand{\bh}{\mathrm{h_b}}

\newcommand{\BDD}{\textsc{Bounded Degree Deletion}}
\newcommand{\EQC}{\textsc{Equitable Coloring}} 
\newcommand{\MMO}{\textsc{Graph Balancing}} 
\newcommand{\CDS}{\textsc{Capacitated Dominating Set}}
\newcommand{\CVC}{\textsc{Capacitated Vertex Cover}}
\newcommand{\MC}{\textsc{Max Cut}} 
\newcommand{\EDS}{\textsc{Edge Dominating Set}}

\newcommand{\OPT}{\mathrm{OPT}}

\begin{document}

\title{Parameterized Approximation Schemes using Graph Widths}
\author{Michael Lampis\thanks{Research partially supported by a Scientific Grant-in-Aid from Ministry of Education, Culture, Sports, Science and Technology of Japan.}}
\institute{Research Institute for Mathematical Sciences (RIMS), Kyoto University \\ \email{mlampis@kurims.kyoto-u.ac.jp}}
\maketitle

\begin{abstract}

Combining the techniques of approximation algorithms and parameterized
complexity has long been considered a promising research area, but relatively
few results are currently known. In this paper we study the parameterized
approximability of a number of problems which are known to be hard to solve
exactly when parameterized by treewidth or clique-width. Our main contribution
is to present a natural randomized rounding technique that extends well-known
ideas and can be used for both of these widths. Applying this very
generic technique we obtain approximation schemes for a number of problems,
evading both polynomial-time inapproximability and parameterized intractability
bounds.

\end{abstract}

\section{Introduction}

Approximation algorithms and parameterized complexity are two of the most
popular ways of dealing with NP-hard optimization problems. Nevertheless, the
two sets of techniques are usually treated independently. It's therefore a very
natural question whether combining the techniques of both theories can be used
to obtain algorithmic results which are out of reach for each one of them
separately.  This has often been identified as a promising research field (see
\cite{Marx08} for a survey), but its development has so far been somewhat
limited.  The goal of this paper is to add some results in this area by
designing parameterized approximation schemes for problems which are both
parameterized intractable (W-hard) and hard to approximate in polynomial time
(APX-hard). 

The problems we will focus on are optimization problems on graphs of bounded
treewidth or clique-width. These two graph widths are of central importance to
parameterized complexity theory. At the same time, they play a significant role
in the design of approximation algorithms, since subroutines employing them are
often used as building blocks of larger algorithms.  Therefore, understanding
the extent to which we can efficiently approximate problems which remain W-hard
for these widths is of potentially great importance from several points of
view.

In this paper we want to show that many such hard problems  actually turn out
to be well-approximable in FPT time.  Perhaps the easiest way to explain our
aim is to state a representative example of the type of results we will
establish.

\begin{theorem}[partial statement] 

There exists a randomized $(1+\eps)$-approximation algorithm for \MC\ running
in time $\left(\frac{\log n}{\eps}\right)^{O(w)} n^{O(1)}$, where $w$ is the
input graph's clique-width.

\end{theorem}

\setcounter{theorem}{0}

\MC\ is of course a problem of central importance in the contexts of both
approximability and parameterized complexity. It is APX-hard (so an
approximation ratio of $1+\eps$ is probably impossible in polynomial time) and
W-hard parameterized by clique-width (so the fastest exact algorithm probably
needs time roughly $n^w$). Our main point here is that using a
\emph{parameterized approximation} approach we can evade these lower bounds,
leading to a $(1+\eps)$-approximation running in time only $(\log n)^{O(w)}$,
that is, an FPT approximation scheme.  More generally, the goal of this paper
is to provide, in a \emph{uniform} way, similar approximation (or bicriteria
approximation) schemes for a diverse set of W-hard and APX-hard graph problems.
The problems for which we will provide algorithms are the following: \MC, \EDS,
\BDD, \CDS, \CVC, \EQC\ and \MMO. For most of these problems we are able to
provide equally efficient algorithms for both treewidth and clique-width (a
detailed description of all results is given further below).

\subsubsection*{Paper overview:} In this paper we adopt a generic technique
that is a variation of the standard dynamic programming used for treewidth and
clique-width.  We observe that for a number of problems which are parameterized
intractable for these widths, the hardness intuitively stems from the fact that
large \emph{integers} need to be stored in the dynamic programming table. These
integers are usually calculated simply by \emph{adding} previously calculated
entries (all the problems listed above fall into this general category, though
for some this is not obvious).  We want to shrink the table, and thus speed up
the algorithms, by storing these integers approximately. 

The basic idea we use is very natural. We fix a parameter $\delta>0$ and
represent all integers in $\{1,\ldots,n\}$ by rounding them to the closest
integer power of $(1+\delta)$. If $\delta$ is not too small
($\delta=\Omega(\frac{1}{\log^c n})$) the natural dynamic programming table's
size is dramatically reduced from $n^w$ to $(\log n )^{O(w)}$. The obvious
obstacle to this approach, however, is that during the process of running a
dynamic programming algorithm on the approximate values the rounding errors
will propagate and potentially pile up to a large error. How can we keep the
errors under control?

In this paper we suggest a very natural randomized rounding approach to this
problem. The main contribution is to show that this rounding idea can be
seamlessly incorporated into the standard dynamic programming techniques of
treewidth and clique-width to give efficient approximation schemes. In order to
make the transition from exact to approximate algorithms as cleanly as possible
we separate the analysis into two parts.  First, we introduce an abstract model
of computation, called Approximate Addition Trees, which captures the essence
of the rounding ideas we described.  We fully analyze the approximation
performance of these Trees and prove some general approximation theorems. Then,
relying on this analysis we give a series of approximation algorithms using
clique-width and treewidth.  It's worth stressing at this point that all the
algorithms we will present follow the standard dynamic programming mold that
should be very familiar to readers accustomed to graph widths. The important
difference is that their analysis, in addition to standard methods, also
crucially relies on our results on Approximate Addition Trees (which we can use
as a black box).  Thus, by abstracting away the Addition Trees, our technique
can be viewed as a natural extension of well-known ideas.  The hope is that
this modularization will allow our technique to be easily reused and eventually
become part of the standard graph width toolkit.

Thus, what is left to describe is the workings of Approximate Addition Trees.
This is, expectedly, the most technical part of the paper. As we will see,
there do exist some important special cases where a complicated analysis can be
avoided (notably, when the input tree is balanced) and there is some value in
these cases since they can help make some algorithms deterministic.  However,
in order to obtain the more interesting results of this paper we need an
analysis of full generality.  In other words, we need to establish an
approximation theorem that works for all Addition Trees without making any
special assumptions about their structure.  Our main technical contribution is
that we do establish such a result and this allows us to analyze all the
algorithms of this paper in terms of Addition Trees. We thus present a robust,
unified technique that works for both treewidth and clique-width (and
potentially other similar graph widths), without relying on any non-trivial
width-specific properties.

\subsubsection*{Summary of results:} Let us now formally state the algorithmic
results presented in this paper. Full problem definitions are given further
below.

In the following theorems, $n^{O(1)}$ factors are omitted from the running
times.

\begin{theorem} \label{thm:cw-main}

Given an $n$-vertex graph $G(V,E)$, a clique-width expression with $w$ labels,
and an error parameter $\eps>0$, there exist randomized algorithms which, with
high probability, achieve the following:

\begin{itemize}

\item Produce an approximate solution to \MC\ with size at least
$\frac{\OPT}{1+\eps}$ in time $(\log n/\eps)^{O(w)}$.

\item Produce an approximate solution to \EDS\ with size at most $(1+\eps)\OPT$
in time $(\log n/\eps)^{O(w)}$.

\item Given an integer $k$, either decide (correctly) that $G$ does not admit
an \EQC\ with $k$ colors or produce a valid $k$-coloring where the ratio of the
sizes of any two color classes is at most $(1+\eps)$ in time $(\log
n/\eps)^{O(k)} 2^{kw} $.

\item Given an integer $\Delta$ find a set of vertices that is at most as large
as the optimal solution for \BDD\ to degree $\Delta$ and whose deletion makes
the maximum degree at most $(1+\eps)\Delta$, in time $(\log n/\eps)^{O(w)}$.

\item Given a capacity for each vertex, find a \CDS\ of size at most $\OPT$,
such that all but at most $\eps n$ vertices are dominated, in time $(\log
n/\eps)^{O(w)}$.

\end{itemize}

In addition, if instead of a clique-width expression we are given a tree
decomposition of width $w$, there exist deterministic algorithms, with the same
running times, achieving all the above. 

\end{theorem}

Furthermore, we also have the following results for treewidth.

\begin{theorem} \label{thm:tw-main}

Given an $n$-vertex graph $G(V,E)$, a tree decomposition of width $w$, and an
error parameter $\eps>0$, there exist deterministic algorithms which achieve
the following: 

\begin{itemize}

\item Produce an approximate solution with cost at most $(1+\eps)\OPT$ for
\MMO, in time $(\log n)^{O(w)}$.

\item Given a capacity for each vertex, find a \CDS\ (or \CVC) of size at most
$\OPT$, such that no capacity is violated by a factor of $(1+\eps)$ or more, in
time $(\log n)^{O(w)}$.

\end{itemize}

\end{theorem}

The algorithms achieving the results of Theorems \ref{thm:cw-main} and
\ref{thm:tw-main} are given in Section \ref{sec:algs}.

\subsubsection*{Previous work:}

Let us first review previous work for the specific problems we will focus on.
In \MC\ we want to partition the vertices of a graph into two sets so that the
number of edges with one endpoint in each set is maximized. \MC\ was shown to
be W-hard when parameterized by clique-width in \cite{FominGLS10}. The problem
is known to be APX-hard in general \cite{papadimitriou1991optimization}.  In
\EDS\ we want to select the smallest possible set of edges such that all edges
share an endpoint with a selected edge. This problem is also APX-hard and
W-hard for clique-width \cite{FominGLS10}. Both problems are FPT parameterized
by treewidth. Let us remark that, in addition to these two problems, very few
other problems are known to be W-hard for clique-width but FPT for treewidth.
The set of problems that behave this way are sometimes considered part of ``the
price of generality'' that clique-width affords, compared to treewidth.
Investigating this price has been an interesting research topic in
parameterized complexity.  One interpretation of the results of this paper is
therefore that this ``price'' is not as high as previously believed, since two
of the most prominent problems from this family can be well-approximated for
clique-width.

In \EQC\ we want to find a proper $k$-coloring of a graph so that all color
classes have the same size.  \EQC\ is known to be W-hard by the results of
Fellows et al.  \cite{fellows2011complexity} even when parameterized by both
the treewidth of the input graph $w$ and the number of colors $k$. An exact XP
algorithm parameterized by $w+k$ can be easily obtained with standard
techniques, but a more general XP algorithm parameterized by $w$ only is shown
in \cite{bodlaender2005equitable}. The problem is easily shown to generalize
\textsc{Graph Coloring} (if we add a sufficiently large independent set to a
graph $G$ then $G$ admits an equitable coloring with $k$ colors if and only if
its chromatic number is at most $k$).

In \MMO\ we are given an edge-weighted graph and need to find an orientation of
the edges so that the maximum weighted out-degree of any vertex is minimized.
\MMO\ is sometimes also called \textsc{Minimum Maximum Outdegree} in the
literature. One motivation for the study of this problem is that it is a
special case of min-makespan scheduling (if vertices represent machines and
edges represent jobs that can be processed only by their endpoints). \MMO\ was
shown to be W[1]-hard parameterized by treewidth by Szeider \cite{szeider11}.
In \cite{AsahiroMO08} it was shown to be in P when all edge weights are equal,
while in \cite{Szeider11MFCS} an XP algorithm was shown when the problem is
parameterized by treewidth.  Regarding polynomial-time approximations, in
\cite{EbenlendrKS08} the problem was shown to be hard to approximate with a
ratio better than 1.5 even if all weights are 1 or 2. In the same paper a
1.75-approximation was given.  

In \BDD\ we want to delete as few vertices as possible to make the maximum
degree of a graph $\Delta$.  \BDD\ was shown to be W[1]-hard parameterized by
treewidth by Betzler et al.~\cite{betzler2012bounded}. The problem generalizes
\textsc{Vertex Cover} and is therefore APX-hard to solve for general graphs.
In \CDS\ (\CVC) we are given a graph where each vertex $v$ has a capacity
$c(v)$ and we want to find a dominating set (resp.~vertex cover) such that each
selected vertex $v$ is used to dominate at most $c(v)$ other vertices
(resp.~cover $c(v)$ edges).  \CDS\ and \CVC\ were shown W[1]-hard when
parameterized by the size of a minimum feedback vertex set (and therefore, also
when parameterized by treewidth) in \cite{DomLSV08}. Both problems are at least
APX-hard to approximate in polynomial time, since they generalize
\textsc{Dominating Set} and \textsc{Vertex Cover}. Since the algorithms we give
approximate the capacity (or degree) constraints they can be described as
$(1,1+\eps)$-bicriteria approximations for these problems.

Let us also recall some more general related work. Very few FPT approximation
schemes are currently known.  For an overview of the most important results see
the survey by Marx \cite{Marx08}.  The same paper gives an FPT approximation
scheme for \textsc{Max Vertex Cover} parameterized by the size of the cover.
This is extended in \cite{SkowronF13} to an FPT approximation scheme for
\textsc{Max Cover}. See also \cite{BonnetP13a} for FPT approximation schemes
for related covering problems.  \textsc{Sum Edge Multicoloring} is a rare
example of a problem currently known to admit an FPT approximation scheme
parameterized by treewidth \cite{Marx04}.

Let us also mention that the notion of FPT approximation also makes sense when
one is trying to obtain constant factor approximations (instead of
approximation schemes) (see e.g. \cite{FratiGGM13}). It's also interesting to
approximate problems which are FPT, if the approximation algorithm can run in
significantly improved time (see \cite{BrankovicF10} or \cite{Bodlaender13}).
The complexity of parameterized approximations for naturally parameterized
problems (that is, parameterized by the size of the solution) has also been
considered.  Unfortunately, the evidence so far seems to suggest that standard
problems, such as \textsc{Clique} and \textsc{Dominating Set} are hard to
approximate even in a parameterized setting
\cite{Escoffier12},\cite{ChitnisHK13}.

In this paper we focus on problems parameterized by treewidth or clique-width.
For an introduction to these notions see
\cite{bodlaender2008combinatorial,courcelle2000linear,EspelageGW01}.  It was
initially believed that problems solvable on trees are almost always FPT
parameterized by treewidth. Gradually, many exceptions were discovered.  This
includes problems such as \textsc{Edge-disjoint Paths} and
\textsc{L(2,1)-coloring}, which are NP-hard for graphs of treewidth 2
\cite{N01,FialaGK05} and \textsc{Steiner Forest} which is NP-hard for graphs of
treewidth 3 (\cite{BateniHM10} gives a PTAS for this problem using treewidth).
More relevant to our purposes are problems which are solvable in polynomial
time for constant treewidth, but not FPT. Some examples of such problems when
parameterized by treewidth (in addition to the problems we consider in this
paper) are the following: \textsc{Target Set Selection} \cite{Ben-ZwiHLN11},
\textsc{Maximum Path Coloring} \cite{Lampis11}, \textsc{List Hamiltonian Cycle}
\cite{Meeks11}, \textsc{List Coloring} (even if parameterized by vertex cover)
\cite{fellows2011complexity}, \textsc{General Factor} \cite{SamerS08},
\textsc{Generalized Satisfiability} parameterized by the treewidth of the dual
or incidence graph \cite{SamerS10}, \textsc{Generalized Domination}
\cite{Chapelle10}, \textsc{Bounded Edge-Disjoint Paths} \cite{GolovachT11}.
All problems which are W-hard for treewidth are of course also W-hard for the
more general clique-width. Additionally, in \cite{FominGLS09} it is shown that
\textsc{Graph Coloring, Hamiltonicity,} \MC\ and \EDS\ are W-hard parameterized
by clique-width.

\section{Definitions and Preliminaries}

We assume that we are using the real-word RAM model. We use $\log(n)$ to denote
the base-2 logarithm of $n$ and $\ln(n)$ to denote the natural logarithm. In
addition $\log_{(1+\delta)}(n)$ is the logarithm base-$(1+\delta)$, for
$\delta>0$. We have $\log_{(1+\delta)}(n) = \ln(n)/\ln(1+\delta)$. The function
$\lfloor x \rfloor$, for $x\in\mathbb{R}$ denotes the maximum integer that is
not larger than $x$. 

We use boldface to denote vectors, for example $\vec{d}$. Sometimes a vector
$\vec{s}\in S^k$ for $S$ a set and $k\in\mathbb{N}$ will also be viewed as a
function from $\{1,\ldots,k\}$ (or some other convenient set of size $k$) to
$S$, and vice-versa. For a function $f:S_1\to S_2$ we use $f^{-1}(u), u\in S_2$
to denote the set $\{ v\in S_1\ |\ f(v)=u \}$.

We will also use the following standard facts.

\begin{lemma} \label{lem:facts}

Let $x,\delta\in\mathbb{R}$. Then the following hold:

\begin{enumerate}

\item \label{item:a} $1+x \le e^x$

\item \label{item:b} If $x\in (0,\frac{1}{2})$ then $e^{x/2} \le 1+x$ 

\item \label{item:c} If $x\in (-\frac{1}{2},\frac{1}{2})$ then $e^{x} \le
1+x+x^2$

\item \label{item:d} If $x \in (-\frac{1}{2},\frac{1}{2})$ then
$|\ln (1+x)| \ge \frac{|x|}{2}$

\item \label{item:e} If $\delta \in (0,\frac{1}{2})$ and $\delta |x| \le
\frac{1}{2}$ then $|(1+\delta)^x-1| \ge \frac{1}{4} \delta |x|$

\end{enumerate}

\end{lemma}

\begin{proof}

For item \ref{item:a} one can consider the function $f(x)=e^x-x-1$. This
function has a global minimum at $x=0$ (this can be established by looking at
its derivative), thus $f(x)\ge f(0)=0$. For item \ref{item:b} we can use the
Taylor expansion $e^{x/2}=\sum_{i=0}^{\infty} \frac{(x/2)^i}{i!} \le 1 +
\frac{x}{2} + \sum_{i=2}^\infty \frac{x^i}{8} \le 1 + \frac{x}{2} +
\frac{x^2}{4} \le 1+x$, where we have used the fact that $x\le \frac{1}{2}$.

For item \ref{item:c} we again use the Taylor expansion $e^x = 1 + x +
\sum_{i=2}^\infty \frac{x^i}{i!} \le 1 + x + \sum_{i=2}^\infty \frac{|x|^i}{2}
\le 1 + x + x^2$, where in the last inequality we used the fact that $|x|\le
\frac{1}{2}$.

For item \ref{item:d} if $x\ge 0$ this follows from item \ref{item:b} by taking
the natural logarithm of both sides. For $x\le 0$ we can use the fact that
since $|x| < 1$ we have $\ln(1-|x|) = - \sum_{i=1}^{\infty} \frac{|x|^i}{i}$,
therefore $|\ln(1+x)| = |\ln(1-|x|)| = \sum_{i=1}^{\infty} \frac{|x|^i}{i} \ge
|x|$.

For item \ref{item:e} first consider the case $x\ge 0$. We have $\ln(1+\delta)
\ge \frac{\delta}{2}$ (by item \ref{item:b}) so since $x\ge 0$ we have
$x\ln(1+\delta) \ge \frac{\delta x}{2} \Rightarrow (1+\delta)^x \ge
e^{\frac{\delta x}{2}} \ge 1+\frac{\delta x}{2}$, where we used item
\ref{item:a}. Thus, $(1+\delta)^x-1 \ge \frac{\delta x}{2} \Rightarrow
|(1+\delta)^x-1| \ge \frac{\delta |x|}{2}$, since $x\ge 0$ and the item is
proved in this case. For $x<0$ we have $|(1+\delta)^x-1| = 1 - (1+\delta)^x =
\frac{(1+\delta)^{|x|}-1}{(1+\delta)^{|x|}}$. Using our calculations from the
$x\ge 0$ case we have $(1+\delta)^{|x|}-1 \ge \frac{\delta|x|}{2}$ so we get
$|(1+\delta)^x-1| \ge \frac{\delta|x|}{2(1+\delta)^{|x|}} \ge
\frac{\delta|x|}{2e^{\delta |x|}} \ge \frac{\delta |x|}{2 \sqrt{e}} \ge
\frac{\delta|x|}{4}$ where we have used that $1+\delta \le e^\delta$ (item
\ref{item:a}) and $\delta|x|\le \frac{1}{2}$. \qed

\end{proof}

We assume the reader has some familiarity with standard definitions from
parameterized complexity, such as the classes FPT and XP (see
\cite{flum2006parameterized}).  For a parameterized problem with input size $n$
and parameter $k$ an FPT Approximation Scheme (FPT-AS) is an algorithm which,
given an error parameter $\eps>0$ runs in time $f(k,\frac{1}{\eps})$ (that is,
FPT time when parameterized by $k+\frac{1}{\eps}$) and produces a solution that
is at most a multiplicative factor $(1+\eps)$ from the optimal (see
\cite{Marx08}).  The problems we consider in this paper are parameterized by
some graph width. We design (randomized) algorithms running in time
$\left(\frac{\log n}{\eps}\right)^{O(k)} n^{O(1)}$. By standard facts in
parameterized complexity theory such running times imply FPT algorithms.

\subsubsection*{Graph widths}

We use standard graph theoretic notation. For an undirected graph $G(V,E)$ and
$S\subseteq V$ we denote by $G[X]$ the graph induced by $X$. We will use the
standard notion of tree decomposition (for an introduction to this notion see
the survey by Bodlaender and Koster \cite{bodlaender2008combinatorial}).  Given
a graph $G(V,E)$ a tree decomposition of $G$ is a tree $T(I,F)$ such that every
node $i\in I$ has associated with it a set $X_i\subseteq V$, called the bag of
$i$.  In addition, the following are satisfied: $\bigcup_{i\in I} X_i = V$; for
all $(u,v)\in E$ there exists $i\in I$ such that $\{u,v\}\subseteq X_i$; and
finally for all $u\in V$ the set $\{ i\in I\ |\ u\in X_i\}$ is a connected
sub-tree of $T$. The width of a tree decomposition is defined as $\max_{i\in I}
|X_i| - 1$. The treewidth of a graph $G$ is the minimum treewidth of a tree
decomposition of $G$.

As is standard, when dealing with problems on graphs of bounded treewidth we
will assume that a ``nice'' tree decomposition of the input graph is supplied
with the input. In a nice tree decomposition the tree $T$ is a rooted binary
tree and each node $i$ of the tree is of one four special types (see
\cite{bodlaender2008combinatorial} for a definition).

We will also use the notion of clique-width (see
\cite{courcelle2000linear,EspelageGW01}).  The set of graphs of clique-width
$w$ is the set of vertex-labelled graphs which can be constructed inductively
using the following operations:

\begin{enumerate}

\item Introduce: $i(l)$, for $l\in\{1,\ldots,w\}$  is the graph consisting of a
single vertex with label $l$.

\item Union: $\cup(G_1,G_2)$, for $G_1, G_2$ having clique-width $w$ is the
disjoint union of these two graphs.

\item Join: $\sigma(G,l_1,l_2)$, for $G$ having clique-width $w$ and
$l_1,l_2\in\{1,\ldots,w\}$ is the graph obtained from $G$ by adding all
possible edges from vertices with label $l_1$ to vertices with label $l_2$.

\item Rename: $\rho(G,l_1,l_2)$, for $G$ having clique-width $w$ and
$l_1,l_2\in\{1,\ldots,w\}$ is the graph obtained from $G$ by changing the label
of all vertices having label $l_1$ to $l_2$.

\end{enumerate} 

A clique-width expression of width $w$ for $G(V,E)$ is a recipe for
constructing a $w$-labelled graph isomorphic to $G$. More formally, a
clique-width expression is a rooted binary tree such that each node has one of
four possible types, corresponding to  the operations described above. In
addition, all leaves are Introduce nodes, each Introduce node has a label
$\in\{1,\ldots,w\}$ associated with it, and each Join or Rename node has two
labels in $\{1,\ldots,w\}$ associated with it. For each node $i$ the graph
$G_i$ is defined as the graph obtained by applying the operation of node $i$ to
the graph (or graphs) associated with its child (or children). All graphs $G_i$
are subgraphs of $G$ and for all leaves with label $l$ we define their
associated graph to be $\eta(l)$.

Again, as is customary, when dealing with a problem parameterized by
clique-width we will assume that a clique-width expression of the input graph
is supplied with the input. We can also assume without loss of generality that
the given clique-width expression has some nice properties. In particular,
whenever the operation $\sigma(G_i,l_1,l_2)$ is used we can assume that there
are no edges between vertices with labels $l_1,l_2$, since otherwise we can
edit the clique-width expression up to this point to remove the Join operations
that produced such edges and this does not affect the final graph.

\section{Approximate Addition Trees} \label{sec:trees}

In this section we describe an abstract model of computation which one may
naturally call Addition Trees. In such a Tree each node calculates a value that
is the sum of the values of its children. We also define an Approximate version
of these trees, where each node \emph{probabilistically} rounds calculated
values to integer powers of $(1+\delta)$, for some parameter $\delta>0$. These
trees capture the rounding scheme that will be the heart of the algorithms of
the next section.  Our goal is to prove that the values of Approximate and
Exact Addition Trees are almost always very close, even if $\delta$ is not too
small (we want $\delta=\Omega(1/\log^cn)$). We require $\delta$ to be in this
range, because in the end the algorithms of the next section will run in time
roughly $(\log n/\delta)^w$. Thus, if we allow $\delta$ to become polynomial in
$n$ (which would make this section easy), we will get algorithms as slow as the
trivial exact ones.

Intuitively, there are two extreme cases to consider here. First, if a tree is
balanced (that is, it has logarithmic height), it is not hard to establish that
rounding errors cannot pile up too badly (Theorem \ref{thm:easy}). Somewhat
surprisingly, this easy case is already sufficient to obtain several
non-trivial algorithmic results, because tree decompositions can always be
balanced reasonably well (more details are given in the next section). However,
to obtain the more interesting results of this paper we need to deal with
clique-width, where the input decomposition cannot in general be balanced.
Therefore, we have to deal with general Addition Trees.

Our proof strategy then is to move on to a second extreme case: caterpillars.
Here the height of the tree is large, but we know that one operand of each
addition has no previously accumulated error. Despite this, this is actually a
pretty hard case. To see why, intuitively one can think of the accumulated
error at each level of the tree as a random variable, since the rounding
performed on each step is randomized. The error has some probability of
increasing and some of decreasing, depending on how we round,  but it changes
by at most a factor of $(1+\delta)$ in each step. So, if we look at its
logarithm (with base $(1+\delta)$) it can (randomly) increase or decrease by at
most 1.  Thus, the process we are trying to analyze is akin to a memoryless
random walk on the real line.  We want to prove that the end result of the walk
after $n$ steps is with high probability contained in an area of size only
roughly $1/\delta = \mathrm{poly}(\log n)$.  Such a statement would be,
however,  false if the walk was completely unbiased, so we cannot rely on
standard tools, such as Chernoff bounds or the Azuma inequality, because the
result they give is too weak (they give concentration in an area of size
roughly $\sqrt{n}$).  Instead, we need to use moment-generating functions to
derive a problem-specific concentration bound that takes into account our
algorithm's tendency to ``self-correct''. Because of this special tendency
(Lemma \ref{lem:self-correct}), our random walk is much more strongly
concentrated around its expectation than usual random walks.

Thus, eventually we establish (in Lemma \ref{lem:paths}) that the approximation
error is small in the caterpillar case, with high probability.  Once this has
been shown we can extend the same ideas to prove a sufficiently good
approximation theorem for general trees by performing induction on the
``balanced height'' of the tree (Theorem \ref{thm:trees}).  Roughly speaking,
the idea is that in any node of an arbitrary tree either both children have
roughly the same accumulated error (in which case the node is balanced) or one
has potentially much higher error (in which case the proof is similar to that
for caterpillars). So, combining the ideas of the two cases we can handle
arbitrary trees.

We remark that the only parts of this section needed to follow the results of
the next one are Definitions \ref{def:AT},\ref{def:AAT} and Theorems
\ref{thm:easy},\ref{thm:trees}. Let us now proceed to give full details.

For the following definition recall that a rooted binary tree is full if all
non-leaf nodes have exactly two children.

\begin{definition} \label{def:AT}

An Addition Tree is a full rooted binary tree $T$ where we associate with each
leaf $l$ a non-negative integer \emph{input} $x_l$ and with each node $v$ a
non-negative integer \emph{value} $y_v$.  The inputs are given with $T$. The
value of each node is calculated as follows: 

\begin{enumerate} 

\item For each leaf $l$ we set $y_l:=x_l$

\item For each internal node $v$ with two children $u_1, u_2$ whose values have
already been calculated we set $y_v:= y_{u_1} + y_{u_2}$.

\end{enumerate}

\end{definition}

Let us also define an approximate version of this model.

\begin{definition} \label{def:AAT}

An Approximate Addition Tree with parameter $\delta$ is a full rooted binary
tree $T$ where we associate with each leaf $l$ a non-negative integer
\emph{input} $x_l$ and with each node $v$ a non-negative \emph{approximate
value} $z_v$. The approximate value of each node is calculated as follows: 

\begin{enumerate}

\item For each leaf $l$ we set $z_l:=x_l$ 

\item For each internal node $v$ with two children $u_1,u_2$ we set $z_v :=
z_{u_1} \oplus z_{u_2}$, where the $\oplus$ operation is defined below.

\end{enumerate} 

Let $a_v:=z_{u_1}+z_{u_2}$. We will call $a_v$ the \emph{initial approximate
value} of $v$. 

We use $\oplus$ to denote the following operation: for two non-negative numbers
$x_1,x_2$ we define $x_1\oplus x_2:=0$ if $x_1=x_2=0$. Otherwise, select a real
number $r\in(0,1)$ uniformly at random and set $x_{1}\oplus x_{2}:=
(1+\delta)^{\lfloor \log_{(1+\delta)}(x_1+x_2) + r \rfloor}$.

\end{definition}

The motivation behind this definition is that the number of possible values
that can be stored in a node is much smaller than in an exact tree.  In
particular, note that whenever $z_v$ is non-zero, it must be an integer power
of $(1+\delta)$. If the maximum value calculated by the exact tree is at most
polynomial in $n$ and $\delta=\Omega(1/\log^c n)$ then there are at most
$\mathrm{poly}(\log n)$ many different values that may appear in an approximate
tree. Using this we will be able to obtain smaller dynamic programming tables
in the next section. First though, we have to show that the approximate values
are close to the correct ones.

Let us now give the definition of approximation error by which we will measure
the progress of our algorithm. Since it is not hard to see that for any node
$v$ for which $y_v=0$ an Approximate Addition Tree will also have $z_v=0$, we
are only concerned with the approximation error for nodes where $y_v\neq 0$.
Therefore, in the remainder we will implicitly assume that we are talking about
a tree where for all $v$, $y_v>0$, because sub-trees with value 0 can be
ignored.

\begin{definition}

Let $v$ be a node of an Addition Tree, $y_v$ its (positive) value and $z_v$ its
approximate value calculated if we view the tree as an Approximate Addition
Tree. Then the error $\lambda_v$ is defined as
$\lambda_v:=\log_{(1+\delta)}\frac{z_v}{y_v}$.

\end{definition}

Note that $\lambda_v$ and $z_v$ are random variables (they depend on the random
rounding choices made during the computation), while $y_v$ is fully specified
once the inputs of the tree are fixed. Informally, $\lambda_v$ measures how
many ``$(1+\delta)$ intervals'' off our approximation is from the correct
interval.

Before we go on, let us make an easy observation that will be sufficient to
handle an important special case.

\begin{theorem} \label{thm:easy}

If an Approximate Addition Tree has maximum depth $h$ then for all nodes $v$ we
always have $|\lambda_v|\le h+1$. Therefore, if $\delta<\frac{\eps}{2h}$ then
for all $v$ we have $\max\{\frac{z_v}{y_v},\frac{y_v}{z_v}\} < 1+\eps$.

\end{theorem}

\begin{proof}

The proof is simple and proceeds by induction on $h$. For trees of height 0
(that is, isolated nodes), $z_v$ has the correct interval, so $|\lambda_v|\le
1$. For the inductive step observe that when adding two values the maximum
absolute relative error  cannot increase by more than a factor of $1+\delta$.
\qed

\end{proof}

As a consequence of Theorem \ref{thm:easy} we get that in trees of height
$O(\log n)$ we can set $\delta = \Theta(1/\log n)$ and get error at most
$1+\eps$ everywhere \emph{with probability 1}. Thus, such balanced trees are an
easy case which is already solved,  even without the use of randomization. As
mentioned though, we also need to handle the much more complicated general
case.

The main intuitive observation that we will use to give an approximation
guarantee can be summarized as follows: the process by which the initial
approximation $a_v$ is calculated is ``self-correcting'', while the rounding
step that follows it is unbiased.  Therefore, the whole process tends to
self-correct.  More precisely, we will show that if
$\lambda_{u_1},\lambda_{u_2}$ are the errors of the two children of a node,
then the error for $a_v$ is somewhere between the two. This means that unless
the two errors are exactly equal, the maximum of the two errors will have a
tendency to decrease.  This intuition will be made more precise in Lemma
\ref{lem:self-correct}. 

First, we need to introduce a definition and an auxiliary lemma.

\begin{definition}

Let $v$ be an internal node of an Approximate Addition Tree, and $a_v$ its
initial approximate value as in Definition \ref{def:AAT}.  We define $p_v:=
\log_{(1+\delta)}(a_v)-\lfloor \log_{(1+\delta)}(a_v)\rfloor$.

\end{definition}

Of course $p_v$ is again a random variable, since it depends on $a_v$.
However, note that $p_v$, as defined, is exactly equal to the probability of
``rounding up'', that is, the probability (over the choice of $r\in(0,1)$) that
$\lfloor \log_{(1+\delta)}(a_v) + r \rfloor > \lfloor \log_{(1+\delta)}(a_v)
\rfloor$.  To see this, observe that we round up if and only if
$\log_{(1+\delta)}(a_v) + r \ge \lfloor \log_{(1+\delta)}(a_v)\rfloor +1
\Leftrightarrow r\ge 1-p_v$ which is an event that has probability $p_v$, since
$r$ is selected uniformly at random and $p_v\in[0,1]$.

\begin{lemma} \label{lem:y2big}

Consider an Approximate Addition Tree with parameter $\delta<\frac{1}{2}$ and
let $v$ be an internal node with two children $u_1,u_2$.  Then $z_{u_2} \ge
\frac{1}{2} \delta p_v z_{u_1}$.

\end{lemma}

\begin{proof}

We have $p_v = \log_{(1+\delta)}(a_v)-\lfloor \log_{(1+\delta)}(a_v)\rfloor$,
but $a_v = z_{u_1}+z_{u_2}$ therefore $\lfloor \log_{(1+\delta)}(a_v)\rfloor
\ge \lfloor \log_{(1+\delta)}(z_{u_1})\rfloor = \log_{(1+\delta)}(z_{u_1})$
where we used the fact that $z_{u_1}$ is an integer power of $(1+\delta)$. 

We now have $p_v \le
\log_{(1+\delta)}(z_{u_1}+z_{u_2})-\log_{(1+\delta)}(z_{u_1}) =
\log_{(1+\delta)}\left( 1 + \frac{z_{u_2}}{z_{u_1}} \right)$. Therefore, to
establish the lemma it is sufficient to show that $\frac{z_{u_2}}{z_{u_1}} \ge
\frac{1}{2}\delta \log_{(1+\delta)}\left( 1 + \frac{z_{u_2}}{z_{u_1}}\right)$.
To ease notation let $\gamma = \frac{z_{u_2}}{z_{u_1}}$.

We have $\gamma \ge \frac{\delta}{2} \log_{(1+\delta)}(1+\gamma)
\Leftrightarrow \gamma \ln(1+\delta) \ge \frac{\delta}{2} \ln(1+\gamma)
\Leftrightarrow (1+\delta)^\gamma \ge (1+\gamma)^{\frac{\delta}{2}}$. Now, from
Lemma \ref{lem:facts}-(\ref{item:b}) and $\delta < \frac{1}{2}$ we have
$(1+\delta)^\gamma \ge e^{\frac{\delta\gamma}{2}}$, while from Lemma
\ref{lem:facts}-(\ref{item:a}) we have $(1+\gamma)^{\frac{\delta}{2}} \le
e^{\frac{\delta\gamma}{2}}$. The result follows. \qed

\end{proof}

We are now ready to give the main self-correction lemma. 

\begin{lemma} \label{lem:self-correct}

Consider an Approximate Addition Tree with parameter $\delta<\frac{1}{2}$. Let
$v$ be an internal node with two children $u_1,u_2$.  If
$\max\{|\lambda_{u_1}|,|\lambda_{u_2}|\}< \frac{1}{4\delta}$ then
$|\log_{(1+\delta)}\frac{a_v}{y_{u_1}+y_{u_2}}| \le
\max\{|\lambda_{u_1}|,|\lambda_{u_2}|\}
- \frac{1}{20} \delta p_v |\lambda_{u_1}-\lambda_{u_2}|$. 

\end{lemma}

\begin{proof}

Informally, what this lemma states is that, before the rounding step is
applied, the absolute value of the error will decrease when compared with the
maximum error up to this point. In fact, the decrease will be proportional to
the absolute difference of the two previous errors.

Since we have not assumed an ordering on the children we can assume without
loss of generality that $|\lambda_{u_1}|\ge|\lambda_{u_2}|$. From the
definitions we have $a_v=(1+\delta)^{\lambda_{u_1}}y_{u_1} +
(1+\delta)^{\lambda_{u_2}}y_{u_2}$. If $\lambda_{u_1}=\lambda_{u_2}$ then the
inequality of the lemma is true, so assume that
$|\lambda_{u_1}|>|\lambda_{u_2}|$.

We now have $\log_{(1+\delta)}\frac{a_v}{y_{u_1}+y_{u_2}} = \lambda_{u_1} +
\log_{(1+\delta)}\frac{y_{u_1}+(1+\delta)^{\lambda_{u_2}-\lambda_{u_1}}y_{u_2}}{y_{u_1}+y_{u_2}}
$. Taking into account that $|\lambda_{u_1}|>|\lambda_{u_2}|$ it is not hard to
see that the second term is positive if and only if the first term is negative.
Therefore, $\left|\log_{(1+\delta)}\frac{a_v}{y_{u_1}+y_{u_2}}\right| =
|\lambda_{u_1}| - \left|
\log_{(1+\delta)}\frac{y_{u_1}+(1+\delta)^{\lambda_{u_2}-\lambda_{u_1}}y_{u_2}}{y_{u_1}+y_{u_2}}\right|$.

Now consider the term $\left|
\log_{(1+\delta)}\frac{y_{u_1}+(1+\delta)^{\lambda_{u_2}-\lambda_{u_1}}y_{u_2}}{y_{u_1}+y_{u_2}}\right|$.
We only need to show that this term is at least as large as $\frac{1}{20}\delta
p_v |\lambda_{u_1}-\lambda_{u_2}|$ to establish the lemma. Observe that this
term is an increasing function of $y_{u_2}$. Therefore, we can use a lower
bound on $y_{u_2}$ to give a lower bound on this term. But by Lemma
\ref{lem:y2big} we have $z_{u_2} \ge \frac{1}{2} \delta p_v z_{u_1}
\Leftrightarrow y_{u_2} \ge \frac{1}{2} \delta p_v
(1+\delta)^{\lambda_{u_1}-\lambda_{u_2}}y_{u_1}$. We thus get:

\begin{eqnarray}
\left| \log_{(1+\delta)}\frac{y_{u_1}+(1+\delta)^{\lambda_{u_2}-
  \lambda_{u_1}}y_{u_2}}{y_{u_1}+y_{u_2}}\right|& \ge&
\left|\log_{(1+\delta)}\frac{y_{u_1}+\frac{1}{2}\delta p_v y_{u_1}}{y_{u_1}+\frac{1}{2}\delta
p_v (1+\delta)^{\lambda_{u_1}-\lambda_{u_2}}y_{u_1}}\right| =  \label{eq:1}\\
&& = \left| \log_{(1+\delta)} \frac{1+\frac{1}{2}(1+\delta)^{\lambda_{u_1}-\lambda_{u_2}}\delta p_v}
  {1+\frac{1}{2}\delta p_v} \right| = \label{eq:2}\\
&& = \frac{\left|\ln\left(1+\frac{\frac{1}{2}\delta p_v\left( (1+\delta)^{\lambda_{u_1}-\lambda_{u_2}}-1\right)}{1+\frac{1}{2}\delta p_v}\right)\right|}{\ln(1+\delta)} \label{eq:3} \ge\\
&& \ge \frac{1}{\delta}\left|\ln\left(1+\frac{\frac{1}{2}\delta p_v\left( (1+\delta)^{\lambda_{u_1}-\lambda_{u_2}}-1\right)}{1+\frac{1}{2}\delta p_v}\right)\right| \ge \label{eq:4}\\
&&  \ge \frac{p_v \left| (1+\delta)^{\lambda_{u_1}-\lambda_{u_2}}-1\right|}{4+2\delta p_v} \ge \label{eq:5}\\
&& \ge \frac{1}{4} \frac{ \delta p_v|\lambda_{u_1}-\lambda_{u_2}|}{4+2\delta p_v} \ge\label{eq:6}\\
&& \ge \frac{1}{20} \delta p_v |\lambda_{u_1}-\lambda_{u_2}|\label{eq:7}
\end{eqnarray}

To go from (\ref{eq:1}) to (\ref{eq:2}) we used the fact that
$|\log(x)|=|\log(1/x)|$, while for (\ref{eq:3}) we simply changed the base of
the logarithm and performed some calculations. For (\ref{eq:4}) we use the fact
that $\ln(1+\delta)\le \delta$ which can be inferred from Lemma
\ref{lem:facts}-(\ref{item:a}). To get (\ref{eq:5}) we use Lemma
\ref{lem:facts}-(\ref{item:d}). To see that Lemma
\ref{lem:facts}-(\ref{item:d}) applies observe that
$(1+\delta)^{\lambda_{u_1}-\lambda_{u_2}} \le (1+\delta)^{\frac{1}{2\delta}}
\le \sqrt{e} \le 2$. To go from (\ref{eq:5}) to (\ref{eq:6}) we use Lemma
\ref{lem:facts}-(\ref{item:e}) and the fact that $|\lambda_{u_1}-\lambda_{u_2}|
\in (-\frac{1}{2\delta},\frac{1}{2\delta})$.  Finally, (\ref{eq:7}) follows
from the fact that $\delta<\frac{1}{2}$ and $p_v\le 1$. \qed

\end{proof}

Lemma \ref{lem:self-correct} will be our main tool in proving that with high
probability the values of an Approximate Addition Tree are not too far from
those of the corresponding exact tree. We will proceed by induction, starting
from a simple case of path-like trees (caterpillars). We need the following
definition:

\begin{definition}

Let $T$ be a rooted full binary tree. The \emph{balanced height} of a node $v$,
denoted $\bh(v)$ is defined as follows:

\begin{enumerate}

\item If $v$ is a leaf then $\bh(v)=0$.

\item If $v$ is an internal node with children $u_1,u_2$ and
$\bh(u_1)>\bh(u_2)$ then $\bh(v)=\bh(u_1)$.

\item If $v$ is an internal node with children $u_1,u_2$ and
$\bh(u_1)=\bh(u_2)$ then $\bh(v)=\bh(v_1)+1$.

\end{enumerate}

\end{definition}

\begin{lemma} \label{lem:bheight}

In any full binary tree with $n$ nodes and root $r$ we have $n\ge 2^{\bh(r)}$.

\end{lemma}

\begin{proof}

The proof proceeds by induction on $n$. The lemma is trivial for $n=1$. For a
larger tree, consider the two trees rooted at children of the root. If they
have the same balanced height then by induction they both have at least
$2^{\bh(r)-1}$ nodes, which means the whole tree has at least $2^{\bh(r)}$
nodes.  Otherwise, one has balanced height $\bh(r)$ and by induction that
sub-tree contains at least $2^{\bh(r)}$ nodes.  \qed

\end{proof}

We will proceed inductively to prove a bound on the probability of a large
error occurring during the computation of an Approximate Addition Tree. Our
bound will depend on $\delta,n$ and the balanced height of the root of the
tree. 

To begin, observe that the case $\bh(r)=0$ is trivial, as this can only occur
if the tree contains only one node. So the first interesting case is
$\bh(r)=1$, which happens if the rooted tree is made up of a single path with a
leaf attached to each internal node and the root. 

In the remainder we will assume we are dealing with an Approximate Addition
Tree with $n$ nodes, root node $r$ and parameter $\delta\in(0,\frac{1}{2})$. We
first establish the following lemma:

\begin{lemma} \label{lem:paths}

If $\bh(r)\le 1$ then for all
$\lambda\in(\frac{2}{\sqrt{\delta}},\frac{1}{4\delta})$ we have $\PR{\exists v
\in T:\ |\lambda_v| > \lambda} \le 2n^2 e^{-\frac{\lambda\sqrt{\delta}}{20}}$.

\end{lemma}

\begin{proof}

We will first bound the probability that a certain node is the first to have
absolute error greater than $\lambda$, that is, the probability that it has
such an error even though all its descendants did not. Then we will take a
union bound over all nodes. Observe that it is sufficient to consider internal
nodes, since the errors for leaves are by definition 0.

Consider an arbitrary internal node $v$. We will bound the conditional
probability that its absolute error is larger than $\lambda$, given the fact
that all its descendants have smaller absolute error. (To ease presentation we
will omit this conditioning, which should be read implicitly in the remainder
of this proof). Let $t\in(0,\frac{1}{2})$ be a parameter we will set later.  We
have 

$$\PR{ |\lambda_v| > \lambda} = \PR{ e^{t|\lambda_v|}
> e^{t\lambda}} \le
 \frac{\E{e^{t|\lambda_v|}}}{e^{t\lambda}}$$

where for the last step we applied Markov's inequality, since we are dealing
with a non-negative random variable.

Let us now try to bound the expectation which appears on the right-hand side.
The node $v$ has two children $u_1,u_2$, one of which must be a leaf, since the
tree has balanced height 1. Without loss of generality let $u_2$ be the leaf,
and therefore $\lambda_{u_2}=0$. By applying Lemma \ref{lem:self-correct} we
get that $|\log_{(1+\delta)}(a_v/y_v)|\le |\lambda_{u_1}| - \frac{1}{10}\delta
p_v |\lambda_{u_1}|$, where we are using the fact that we are conditioning
under the event that $|\lambda_{u_1}|<\lambda \le \frac{1}{4\delta}$.

Let us now look at the rounding step. Recall that we denote by $p_v$ the
probability of ``rounding up''. We have that $\log_{(1+\delta)}(z_v) =
\log_{(1+\delta)}(a_v) + r_v$, where $r_v$ is a random variable that depends on
the choice of $r\in(0,1)$. In particular, $r_v$ has the following distribution:
with probability $p_v$ we have $\log_{(1+\delta)}(z_v) = \lfloor
\log_{(1+\delta)}(a_v) \rfloor + 1$ and therefore $ r_v = 1 - p_v$.  On the
other hand, with probability $1-p_v$ we have $\log_{(1+\delta)}(z_v) = \lfloor
\log_{(1+\delta)}(a_v) \rfloor$ and therefore $ r_v = -p_v$.

Rewriting we get $\lambda_v = \log_{(1+\delta)}\frac{z_v}{y_v} =
\log_{(1+\delta)}\frac{a_v}{y_v} + r_v \Rightarrow |\lambda_v| =
|\log_{(1+\delta)}\frac{a_v}{y_v} + r_v| = |\log_{(1+\delta)}\frac{a_v}{y_v}| +
r'_v \le |\lambda_{u_1}| - \frac{1}{10}\delta p_v |\lambda_{u_1}| + r'_v$ where
we define $r'_v$ as follows: $r'_v:= r_v$ if $a_v\ge y_v$ and $r'_v:=-r_v$
otherwise. Let us now bound the change caused in the expectation by the
rounding step.

\begin{claim} 

For $t>0$ we have $\E{e^{t r'_v}\ |\ p_v}\le e^{t^2p_v }$

\end{claim}

\begin{proof}

In case $r'_v=r_v$ we have $\E{e^{t r_v}|p_v} = p_v e^{t(1-p_v)} + (1-p_v)
e^{-tp_v} = e^{-tp_v}(1-p_v+p_v e^t) \le e^{-tp_v} e^{-p_v+p_ve^t} \le e^{-tp_v
- p_v + p_v + tp_v + t^2p_v} = e^{t^2p_v }$, where we used Lemma
\ref{lem:facts}-(\ref{item:a}) and Lemma \ref{lem:facts}-(\ref{item:c}).

In case $r'_v=-r_v$ we have $\E{e^{-t r_v}|p_v} = p_v e^{-t(1-p_v)} + (1-p_v)
e^{tp_v} = e^{tp_v}(1-p_v+p_v e^{-t}) \le e^{tp_v} e^{-p_v+p_ve^{-t}} \le
e^{tp_v -p_v + p_v -tp_v +t^2p_v} = e^{t^2p_v}$ where again we used Lemma
\ref{lem:facts}. \qed

\end{proof}

We now have:

\begin{eqnarray}
 \E{e^{t|\lambda_v|}} &=& \E{e^{t|\lambda_v|}\ |\  |\lambda_{u_1}|\ge \frac{1}{\sqrt{\delta}}} \PR{|\lambda_{u_1}|\ge \frac{1}{\sqrt{\delta}} } 
  + \E{e^{t|\lambda_v|}\ |\ |\lambda_{u_1}|< \frac{1}{\sqrt{\delta}}} \PR{|\lambda_{u_1}|<
\frac{1}{\sqrt{\delta}}} \label{eq:8}
\end{eqnarray}

The second term of the right-hand-side of (\ref{eq:8}) can be upper-bounded by
$e^{t+t\sqrt{1/\delta}}$, since $|\lambda_v|\le |\lambda_{u_1}|+1$. We now seek
to find a value of $t$ such that the first term is upper-bounded by
$\E{e^{t|\lambda_{u_1}|}}$.

\begin{eqnarray}
 \E{e^{t|\lambda_v|}\ |\ |\lambda_{u_1}|\ge \sqrt{1/\delta}} &\le& 
  \E{e^{t|\lambda_{u_1}| - t\frac{\delta}{20} p_v |\lambda_{u_1}| + t r'_v}\ |\ |\lambda_{u_1}|\ge \sqrt{1/\delta}} \le
\label{eq:9}\\
 &\le& \E{\E{e^{t|\lambda_{u_1}| - t\frac{\sqrt{\delta}}{20} p_v + t
r'_v}\ |\ p_v}\ |\ |\lambda_{u_1}| \ge \sqrt{1/\delta}} \le \label{eq:10}\\ 
&\le& \E{e^{t|\lambda_{u_1}| - t\frac{\sqrt{\delta}}{20} p_v
 + t^2p_v}\ |\ |\lambda_{u_1}| \ge \sqrt{1/\delta}} \label{eq:11}
\end{eqnarray}

To obtain (\ref{eq:9}) we used Lemma \ref{lem:self-correct}. Then to obtain
(\ref{eq:10}) we used the fact that we are conditioning on
$|\lambda_{u_1}|\ge\sqrt{1/\delta}$ to replace $\lambda_{u_1}$ in one term of the
exponent. We also used standard properties of conditional expectations to
remove the dependence on $r'_v$, which using the claim gives (\ref{eq:11}). We
can now set $t=\frac{\sqrt{\delta}}{20}$ (so we do indeed have
$t\in(0,\frac{1}{2})$ as promised) and we get 

\begin{eqnarray*}
 \E{e^{\frac{\sqrt{\delta}}{20}|\lambda_v|}} &\le& 
   \E{e^{\frac{\sqrt{\delta}}{20}|\lambda_{u_1}|}\ |\ |\lambda_{u_1}|\ge \sqrt{1/\delta}}\PR{|\lambda_{u_1}|\ge \sqrt{1/\delta}}
   + e^{\frac{\sqrt{\delta}}{20}+\frac{1}{20} } \le \\
 &\le& \E{e^{\frac{\sqrt{\delta}}{20}|\lambda_{u_1}| }} + 2
\end{eqnarray*}

Since we have established the above for any internal node $v$, it now follows
by induction than for any node $v$ we have

$$ \E{e^{\frac{\sqrt{\delta}}{20}|\lambda_v|}} \le 2n $$

Using this fact, we can now conclude that the probability that an arbitrary
internal node $v$ is the first to have absolute error greater than $\lambda$ is
at most $2ne^{-\frac{\lambda\sqrt{\delta}}{20}}$. Therefore, by union bound,
the probability that some node has absolute error at least $\lambda$ is at most
$2n^2e^{-\frac{\lambda\sqrt{\delta}}{20}}$, as claimed. \qed

\end{proof}

We are now ready to extend the previous lemma to trees of arbitrary balanced
height.

\begin{lemma} \label{lem:trees}

Let $T$ be an Approximate Addition Tree with root $r$ and $\bh(r)= h$. Then for
all $\lambda\in(\frac{2}{\sqrt{\delta}},\frac{1}{4\delta})$ we have $\PR{
\exists v\in T\ :\ |\lambda_v| > \lambda h} \le (h+1)n^{h+1}
e^{-\frac{\lambda\sqrt{\delta}}{20}}$. 

\end{lemma}

\begin{proof}

The proof proceeds by induction on $h$. The base case for $h=1$ is exactly
Lemma \ref{lem:paths}, which we have already established. Now, assume that the
result has been proved for trees of balanced height up to $h-1$. Let $B_{h-1}$
be the event that there exists some node $u$ with $\bh(u)\le h-1$ such that
$|\lambda_u|>(h-1)\lambda$. Then we have:

$$ \PR{ \exists v\in T\ :\ |\lambda_v| > \lambda h} \le 
 \PR{ B_{h-1} } + 
\PR{ \exists v\in T\ :\ |\lambda_v| > \lambda h\ |\ \neg B_{h-1}} $$ 

Let us bound the two terms separately. For the first term, for every node $u\in
T$ with $\bh(u)\le h-1$ we consider the sub-tree rooted at $u$ and containing
all its descendants. By the inductive hypothesis in a rooted tree of balanced
height at most $h-1$ the probability that there exists a vertex $u$ with
$|\lambda_u|>(h-1)\lambda$ is at most
$hn^he^{-\frac{\lambda\sqrt{\delta}}{20}}$. There are at most $n$ such trees
considered, so by union bound $\PR{ B_{h-1} } \le
hn^{h+1}e^{-\frac{\lambda\sqrt{\delta}}{20}}$.

For the second term, the argument is similar to that of the proof of Lemma
\ref{lem:paths}. Consider an arbitrary node $v$ of the tree such that
$\bh(v)=h$. We will bound the probability that this is the first node with
absolute error greater than $\lambda h$, given that the absolute errors of all
its descendants with the same balanced height are strictly smaller than
$\lambda h$.  This time, we are also conditioning on the event $\neg B_{h-1}$,
so its descendants with smaller balanced height have error at most
$(h-1)\lambda$.

Observe that any node $v$ of balanced height $h$ has at most one child of
balanced height $h$ (otherwise the balanced height of the tree would be at
least $h+1$). If both children of $v$ have balanced height $h-1$ and the event
$\neg B_{h-1}$ is true, then $|\lambda_v|\le (h-1)\lambda+1 < \lambda h$ with
probability 1. So the interesting case is when exactly one child of $v$ has
balanced height $h$.

Let $u_1,u_2$ be the two children of $v$ with $\bh(u_1)=h$ and $\bh(u_2)<h$.
Again, $t\in(0,\frac{1}{2})$ is a parameter to be fixed later and we are
conditioning over the events that all descendants of $v$ have absolute error at
most $\lambda h$ and $\neg B_{h+1}$ is true.

$$\PR{ |\lambda_v| > \lambda h} = \PR{ e^{t|\lambda_v|}
> e^{t\lambda h}} \le
 \frac{\E{e^{t|\lambda_v|}}}{e^{t\lambda h}}$$

We now need to bound the expectation which appears on the right-hand side.
Similarly to Equation (\ref{eq:8}) of Lemma \ref{lem:paths} we have:

\begin{eqnarray}
 \E{e^{t|\lambda_v|}} &=& \E{e^{t|\lambda_v|}\ |\  |\lambda_{u_1}|\ge (h-1)\lambda + \sqrt{1/\delta}} \PR{|\lambda_{u_1}|\ge (h-1)\lambda + \sqrt{1/\delta}} \nonumber \\
 && + \E{e^{t|\lambda_v|}\ |\ |\lambda_{u_1}|< (h-1)\lambda + \sqrt{1/\delta}} \PR{|\lambda_{u_1}|< (h-1)\lambda+
\sqrt{1/\delta}} \label{eq:12}
\end{eqnarray}

By using the fact that $\neg B_{h-1}$ is true we know that if
$|\lambda_{u_1}|>(h-1)\lambda + \frac{1}{\sqrt{\delta}}$ then the difference
$|\lambda_{u_1}-\lambda_{u_2}|\ge\frac{1}{\sqrt{\delta}}$. Thus, we can
upper-bound the first term by $\E{e^{t|\lambda_{u_1}|}}$ in exactly the same
way as in Lemma \ref{lem:paths} by setting $t=\frac{\sqrt{\delta}}{20}$.  The
second term is at most $e^{t((h-1)\lambda+\frac{1}{\sqrt{\delta}}+1)}$. We can
now use the same argument as in Lemma \ref{lem:paths} to get

\begin{eqnarray*}
\E{e^{t|\lambda_v|}} \le n e^{\frac{\sqrt{\delta}}{20}((h-1)\lambda+\frac{1}{\sqrt{\delta}}+1)} &\Rightarrow& 
\PR{ |\lambda_v| > \lambda h} \le 2n e^{-\frac{\lambda\sqrt{\delta}}{20}}
\end{eqnarray*}

The above is an upper bound on the probability that an arbitrary node $v$ of
balanced height $h$ is the first to have absolute error more than $\lambda h$,
assuming all nodes of balanced height at most $h-1$ have error at most
$(h-1)\lambda$. By union bound, the probability that any node of balanced
height $h$ has error more than $\lambda h$ (again conditioned on $\neg
B_{h-1}$) is at most $2n^2e^{-\frac{\lambda\sqrt{\delta}}{20}}$. Summing with
the upper bound we have on the probability of $B_{h-1}$ we have that the
probability of any node having absolute error more than $\lambda h$ is at most
$ hn^{h+1}e^{-\frac{\lambda\sqrt{\delta}}{20}} +
2n^2e^{-\frac{\lambda\sqrt{\delta}}{20}} <
(h+1)n^{h+1}e^{-\frac{\lambda\sqrt{\delta}}{20}}$ and the result follows. \qed

\end{proof}

We are now ready to prove the main theorem of this section

\begin{theorem} \label{thm:trees}

Let $T$ be an Approximate Addition Tree on $n$ nodes with parameter
$\delta\in(0,\frac{1}{2})$.  There exists a fixed constant $C>0$ such that for
all $\eps\in(0,\frac{1}{8})$ and sufficiently large $n$, if
$\delta<\frac{\eps^2}{C\log^6n}$ we have

$$ \PR{\exists v\in T: \max\{\frac{z_v}{y_v},\frac{y_v}{z_v}\} > 1+\eps} \le
n^{-\log n}$$

\end{theorem}

\begin{proof}

The theorem follows from Lemma \ref{lem:trees} as follows. First, note that
$|\lambda_v|>\lambda h \Leftrightarrow \max\{\frac{z_v}{y_v},\frac{y_v}{z_v}\}
> (1+\delta)^{\lambda h}$. By Lemma \ref{lem:facts} $(1+\delta)^{\lambda h} \ge
 e^{\frac{\delta \lambda h}{2}} \ge 1 + \delta \lambda h/2$.

Now we set $\lambda=\frac{2\eps}{\delta h}$. Notice that $\frac{2\eps}{\delta
h}< \frac{1}{4\delta}$ because $h\ge 1$ (otherwise the tree is trivial). Also,
$\frac{2\eps}{\delta h} > \frac{2}{\sqrt{\delta}} \Leftrightarrow \eps > h
\sqrt{\delta}$ which holds because by Lemma \ref{lem:bheight} we have $h\le
\log n$. Therefore, we can apply Lemma \ref{lem:trees} for the chosen
$\lambda$. We get

$$ \PR{ \exists v\in T: \max\{\frac{z_v}{y_v},\frac{y_v}{z_v}\} > 1 + \eps} \le
(h+1) n^{h+1} e^{-\frac{\eps}{10h\sqrt{\delta}}} $$

The result follows by noting that the bound on the right is an increasing
function of $h$ and $h\le \log n$. \qed

\end{proof}

\section{Approximation Schemes} \label{sec:algs}

We are now ready to use the results of the previous section to design some
approximation algorithms. As mentioned, we only need Definitions
\ref{def:AT},\ref{def:AAT} (including the definition of the $\oplus$ operation)
and Theorems \ref{thm:easy}, \ref{thm:trees}.  Let us first describe the
general technique.

In all problems we will assume we are supplied either a tree decomposition of
width $w$ or a clique-width expression with $w$ labels. An important property
of all the problems we consider is that they admit an exact dynamic program
that takes time (roughly) $n^w$. These dynamic programs calculate a set of $w$
integers on each node of the decomposition by adding previously calculated
values or values read from the graph. The idea is to reuse this dynamic
program, but round all values to integer powers of $(1+\delta)$ and perform all
additions with the $\oplus$ operation. Note that, we are in fact able to also
handle some other auxiliary operations besides addition (such as comparisons).
It will, however, be important that our dynamic programs avoid using
subtractions (because then we lose the approximation guarantee) and part of our
effort is devoted to achieving this.

How do we analyze the approximation ratio of such an algorithm? As usual, each
node of the tree decomposition or the clique-width expression represents a
subgraph of the original graph. For each problem, we have a notion of partial
solution confined to a subgraph. We want to show that for each partial solution
there exists (with high probability) an entry in the approximate dynamic
programming table that matches its value and vice-versa. This is where Addition
Trees become useful. For each partial solution value we define inductively an
Addition Tree $T$ whose root calculates that value. We then inductively
establish that the approximate dynamic programming table contains some value
that \emph{follows the same distribution} as the result of $T$, if $T$ is
viewed as an Approximate Addition Tree. Thus, the approximate dynamic
programming table contains (with high probability) an almost correct value. To
simplify things, by tweaking the parameters appropriately we can make the
probability high enough that \emph{all} the values of the approximate table are
close to being correct, assuming $w$ is not too large.  Observe that this
allows the analysis to be performed using essentially the same inductive
reasoning as for standard (exact) dynamic programming algorithms.

What remains to say is what values we select for the parameter $\delta$. Here
we have two choices. In the general case, we set $\delta$ to the value dictated
by Theorem \ref{thm:trees}. This works quite well essentially all the time, as
we can guarantee that with high probability the Addition Trees we consider
during the analysis of our algorithm will have almost correct values.  We can
thus obtain randomized approximation schemes for both clique-width and
treewidth with the promised running time of roughly $(\log n/\eps)^{O(w)}$.
This is the general technique we consider to be the main contribution of this
paper.

Nevertheless, as mentioned there exists an important special case where things
can become simpler. It is known that for any graph of treewidth $w$ there
exists a tree decomposition of width $O(w)$ and only logarithmic (in $n$)
height \cite{BodlaenderH98}. 
Thus, another approach available to us is to use this theorem to first balance
the decomposition and then rely on Theorem \ref{thm:easy}, instead of the more
general Theorem \ref{thm:trees}.  Unfortunately, this does not speed up our
algorithms significantly (because the larger $\delta$ given in Theorem
\ref{thm:easy} is almost out-weighed by the blow-up in the width of the
decomposition). Importantly, it is impossible to apply this approach to
clique-width, because similar balancing results blow up the number of labels
used exponentially (\cite{CourcelleK07}), which would result in an unreasonable
running time of roughly $(\log n)^{2^w}$ (and this is likely to be
unavoidable). More generally, relying on a treewidth-specific balancing theorem
detracts from the bigger point of developing a general technique, since there
are many graph widths for which balancing would not even make sense
(e.g.~pathwidth).  Despite these shortcomings, we still believe there may be
some value in this treewidth-specific balancing trick and mention it as an
alternative approach because it allows us to get rid of randomization in this
case. To the best of our knowledge, this is the first application of tree
decomposition balancing theorems to approximation algorithms (the main
motivation for their study so far was parallel and distributed algorithms). It
would be interesting to see if in the future a combination of a balancing
preprocessing step with the ideas of Theorem \ref{thm:trees} could lead to
better running times or more derandomized algorithms.

In the rest of this section we first, for the sake of completeness, give a
complete proof of one algorithmic result, namely the approximation scheme for
\MC. We then simply sketch the dynamic programs for all other problems, since
the analysis is similar, highlighting some interesting details. We also give
some brief remarks explaining why some algorithms cannot (easily) be extended
to clique-width and why for some problems we only get bi-criteria
approximations.

\subsection{Max Cut}

In this section we attack the \MC\ problem parameterized by clique-width. In
\MC\ we are given a graph $G(V,E)$ and are asked to find a partition of $V$
into $L,R\subseteq V,\ L=V\setminus R$ such that the number of edges with
exactly one endpoint in $L$ is maximized. We assume that a clique-width
expression for $G$ with $w$ labels is supplied as part of the input.

We will follow the straightforward dynamic program for this problem, as used
for example in \cite{FominGLS10}. The idea is to describe a partial solution by
keeping track of the intersection of each label set with $L$. The first
difference is of course that we will round all values to save time. The second
(somewhat counter-intuitive) difference is that for each label set we will keep
track of the size of its intersection with \emph{both} $L$ and $R$. In the
exact program this would be wasteful, since one of these two numbers can be
calculated from the other, using the (known) size of the whole label set.
However, here we are trying to implement a dynamic program that relies only on
additions. As a side effect, our dynamic program will likely include solutions
which are clearly incorrect (the sum of the sizes of the two intersections
being different from the size of the set) but this will not be a major problem
if we can guarantee that all values are at most $(1+\eps)$ far from a valid
solution.

\subsubsection*{Dynamic Program} 

As mentioned, we view the clique-width expression as a rooted binary tree.
First, define the set $B:=\{0\}\cup \{ (1+\delta)^j\ |\ j\in\mathbb{N}\}$.
Informally, $B$ is the set of rounded values that may appear in our table. Even
though $B$ is infinite, whenever the algorithm produces an entry with value
larger than $(1+\eps)n^2$ we simply drop that entry. It will not be hard to see
that this will not affect the analysis if all entries are within a $(1+\eps)$
factor of being correct (then nothing will be dropped). Observe that the size
of $B$ then becomes $\log_{(1+\delta)}(n^2) = \mathrm{poly}(\log n/\eps)$, if
we set $\delta$ according to Theorem \ref{thm:trees}.

The dynamic programming table $D_i$ for a node $i$ is a subset of $B^w\times
B^w\times B$. The informal meaning is that an entry $(\vec{l},\vec{r},c)\in
D_i$ if and only if there exists a partition of the subgraph of $G$ represented
by the sub-tree rooted at $i$ into $L_i,R_i$ such that:

\begin{enumerate}

\item for all $l\in\{1,\ldots,w\}$ the number of vertices in $L_i$
(respectively $R_i$) with label $l$ is roughly $\vec{l}(l)$ (respectively
$\vec{r}(l)$)

\item the number of edges with exactly one endpoint in $L_i$ is roughly $c$

\end{enumerate}

A dynamic programming algorithm can now be formulated in a straightforward way.
It is easy to fill the table for initial nodes. For Rename and Union nodes the
exact dynamic program would perform some addition, which we now replace with
the $\oplus$ operation. For example, consider a Rename node  with labels
$l_1\to l_2$. For each entry $(\vec{l},\vec{r},c)$ in the child's table we
create an entry $(\vec{l'},\vec{r'},c)$ in the current node as follows: we set
$\vec{l'}(l_1):=0$, $\vec{l'}(l_2):= \vec{l}(l_1)\oplus \vec{l}(l_2)$ and
$\vec{l'}$ the same as $\vec{l}$ for other labels to make the vector $\vec{l'}$
of a new entry (similarly for $\vec{r'}$).  In the same way, for a Union node,
for each entry $(\vec{l_1},\vec{r_1},c_1)$ in the first child's table and for
each entry $(\vec{l_2},\vec{r_2},c_2)$ in the second child's table we construct
an entry $(\vec{l_1}\oplus \vec{l_2}, \vec{r_1}\oplus\vec{r_2}, c_1\oplus
c_2)$, where $\oplus$ is applied component-wise for vectors.

For Join nodes with labels $l_1,l_2$ we do something similar. For each entry
$(\vec{l},\vec{r},c)$ in the child's table construct a new vector with the same
$\vec{l},\vec{r}$ and $c':= c \oplus (\vec{l}(l_1)\cdot \vec{r}(l_2) +
\vec{l}(l_2)\cdot \vec{r}(l_1))$.  Note that if the elements of
$\vec{l},\vec{r}$ are known with error at most $(1+\eps)$ then the second term
of this addition is known with error at most $(1+\eps)^2\approx 1+2\eps$.

\subsubsection*{Analysis}

First, observe that because the running time of the algorithm is polynomial in
the size of the tables, the algorithm clearly achieves the running time stated
in Theorem \ref{thm:cw-main}. We only have to prove its approximation
guarantee.

Any node $i$ of the clique-width expression defines a subgraph $G_i$ of $G$
(the graph produced by the sub-expression rooted at $i$). A partial solution is
simply a restriction of a solution $L,R$ for $G$ to the vertices of $G_i$. The
\emph{signature} of a partial solution is a vector of $2w+1$ integers that
would represent this solution in an exact dynamic programming table. In other
words, the signature is the entry we would expect to see in the table $D_i$ if
we were not rounding. In this case, it contains the exact size of the
intersections of $L,R$ with each label set and the size of the cut in $G_i$. To
keep things simpler, we will only be concerned with the $2w$ values that
represent the intersection. As mentioned, if we can get these right, the
algorithm also gets the size of the cut right within a factor of roughly
$(1+2\eps)$.

Consider a partial solution and its signature. Our strategy is to inductively
define a mapping that gives for each such signature a collection of $2w$
Addition Trees. The exact results of these trees will be the values of the
signature (that is, the sizes of the intersections of the two sets with each
label). We will then establish (by induction) that there exists an entry in our
algorithm's table whose values follow the same distribution as those Trees, if
they are viewed as Approximate Addition Trees.  It will follow that for every
partial solution there exists, with high probability, an entry with roughly the
same values. The converse statement can be shown with essentially the same
ideas.

The above can clearly be done for initial nodes, where all values stored are 0
or 1, so our algorithm stores them exactly. The relevant Addition Trees are of
height 0. Assume inductively that we have established the correspondence up to
a certain height in the clique-width expression. Let us then consider a Rename
node $i$ with labels $l_1\to l_2$ and a child $j$.  Consider a partial solution
to $G_i$.  This partial solution has some corresponding partial solution in
$G_j$.  By the inductive hypothesis, there exists an entry in
$(\vec{l},\vec{r},c)\in D_j$ corresponding to this solution.  In particular,
there exist Approximate Addition Trees whose results have the same distribution
as $\vec{l}(l_1)$ and $\vec{l}(l_2)$ and whose exact values are the same as the
corresponding values in the partial solution to $G_j$.  Consider the Tree
obtained from these by adding a new root and making the old roots its children.
This Tree now follows the same distribution as the value $\vec{l'}(l_2)$
calculated by our algorithm.  Its exact value is the value we would get in the
exact signature (since the same was true for the sub-trees).  Reasoning in the
same way about the other coordinates we have completed the inductive step in
this case.

The cases of Join and Union nodes can be handled with similar inductive
arguments as above. We can thus establish that for any valid partial solution
signature there exists an entry of the table that approximately corresponds to
it, in the sense that their respective values are connected through Addition
Trees. It is also not hard to use inductive reasoning to also establish the
converse (for every approximate entry there exists a corresponding partial
solution).  We now select the entry in the root's table that has maximum $c$.
With high probability, it must correspond to a solution with approximately the
same cut size (otherwise we can find an Approximate Addition Tree with large
error). Retracing the steps of the dynamic programming we can turn this entry
into an actual cut.

\subsection{Edge Dominating Set}

Let us now give a dynamic programming algorithm for \EDS. Here, things are a
little trickier, because it's not immediately obvious that using subtractions
can be avoided. We will use the following equivalent version of the problem: we
are looking for a minimum-size set of vertices $S$ such that $S$ is a vertex
cover of $G$ and $G[S]$ has a perfect matching. It is not hard to see that this
is the same problem (intuitively, $S$ is the set of vertices incident on an
edge of the edge dominating set) because an optimal edge dominating set is also
a maximal matching.

We define the dynamic programming table $D_i$ for a node $i$ as a subset of
$(B\cup\{F\})^w\times B^w$, where $F$ is a special symbol.  Let $G_i$ be the
corresponding subgraph.  Fix a solution to the problem in $G$, that is a vertex
cover $S$ and a matching of its vertices $M$.  The intended meaning is that an
entry $(\vec{s},\vec{c})\in D_i$ represents this solution if

\begin{enumerate}

\item for all $l\in\{1,\ldots,w\}$ the number of vertices of $G_i$ with label
$l$ included in $S$ is roughly $\vec{s}(l)$. If $\vec{s}(l)=F$ then \emph{all}
vertices with label $l$ are in $S$.

\item for all $l\in\{1,\ldots,w\}$ the number of vertices of $G_i$ with label
$l$ incident to an edge of $M\cap G_i$  is roughly $\vec{c}(l)$

\end{enumerate}

Informally, $\vec{s}$ tells us how many vertices we have selected in $S$ from
each label set, while $\vec{c}$ tells us how many of the selected vertices have
already been matched. Clearly, the intended meaning implies that $\vec{c}(j)\le
\vec{s}(j)$ for all $j$, but this may not always be the case. We will still be
happy if this is true up to an error factor $(1+\eps)$. In the calculations
below, when we perform arithmetic operations on a value $\vec{s}(l)$ that is
equal to $F$ we substitute it with the size of the set $V_l$ of vertices with
label $l$.

Let us now describe the algorithm. Initial nodes are easy ($\vec{c}$ is
all-zero, $\vec{s}$ can have a single coordinate that is 0 or F). For Rename
and Union nodes we just have to add (component-wise) appropriate entries,
taking care to maintain $F$ values if possible.  The interesting case is Join
nodes.

Let $i$ be a join node with labels $l_1 \leftrightarrow l_2$ and child $j$.
For each entry $(\vec{s},\vec{c})\in D_j$ do the following. Let $V_{l_1},
V_{l_2}$ be the set of vertices with label $l_1,l_2$ respectively in $G_i$. If
$\vec{s}(l_1) \neq F$ and $\vec{s}(l_2) \neq F$ then ignore this entry because
this partial solution is not a vertex cover of $G_i$.  Otherwise, for each
$m\in\{0,\ldots,\min(|V_{l_1}|, |V_{l_2}|)\}$ calculate a vector $\vec{c_m}$
which is identical to $\vec{c}$ except that $\vec{c_m}(l_1)=\vec{c}(l_1)\oplus
m$ and $\vec{c_m}(l_2)=\vec{c}(l_2)\oplus m$. Informally, we are selecting how
many of the join edges will eventually be used in the matching $M$. If we have
$\vec{c_m}(l_1)>(1+\eps)\vec{s}(l_1)$ or $\vec{c_m}(l_2)>(1+\eps)\vec{s}(l_2)$
ignore $\vec{c_m}$. Otherwise, add $(\vec{s},\vec{c_m})$ to $D_i$.

Once the root's table has been calculated we select the entry
$(\vec{s},\vec{c})$ such that $\sum_l \vec{s}(l)$ is minimized, among entries
where $\vec{c}(l) \ge \vec{s}(l)/(1+\eps)$. Retracing the steps of the dynamic
programming we then obtain a vertex cover $S$.  In polynomial time we can
calculate a maximum matching in $G[S]$.  This is our initial candidate
solution. If some vertex of $S$ is unmatched we add one of the edges connecting
it to $V\setminus S$. We now have an edge dominating set.

We are now faced with two problems. First, we need to prove that $S$ has
roughly the same number of vertices as an optimal solution. Second, we need to
prove that $G[S]$ contains an almost perfect matching (more precisely, it
contains a set of edges $M$ such that almost all vertices are incident on one
edge of $M$).  We can again rely on an inductive analysis using Approximate
Addition Trees, as in the case of \MC.  The main difference is that our
algorithm now occasionally drops some partial solutions.  This happens in the
case of Join nodes when we have not selected enough vertices of
$V_{l_1},V_{l_2}$ to produce a proper vertex cover.  This step is easily seen
to be correct, as there is no approximation involved.  Alternatively, solutions
are dropped when we are trying to select too many of the new edges in $M$.  In
this case, we leave enough ``slack'' in our comparisons so that if a solution
is erroneously dropped we can extract an Approximate Addition Tree with high
error, something that can only happen with low probability.  Barring this, we
can assume that all partial solutions are represented in the root's table and
each table entry approximately corresponds to a solution.  Thus, the solution
we select in the end will have an almost optimal size for $S$ and contain an
almost perfect matching.

\subsection{Equitable Coloring}

This problem admits a very simple additive dynamic program. Let us briefly
describe it. In the case of clique-width, for each node the table is a subset
of $(B^k)^w$. Informally, the signature of a partial coloring is the number of
vertices of each color contained in each label set. We can guarantee to produce
a valid coloring if we make sure that in all Join nodes we drop partial
solutions that use the same color somewhere in both label sets. This can be
done without a problem because $0$ values are stored exactly in our table. In
the end, we pick the table entry appearing to give the most equitable coloring
and extract a coloring by retracing the dynamic program. The running time is
$(\log n/\eps)^{O(kw)}$.

Let us also note that we can give a faster algorithm for treewidth. Here the
dynamic program needs to remember the size of each color class and the coloring
of the $w$ vertices of a bag. Thus, the table is a subset of $B^w\times
\{1,\ldots,k\}^w$, leading to a running time of $(k\log n/\eps)^{O(w)}$.

\subsection{Capacitated Dominating Set and Vertex Cover}

Let us first describe the algorithm for \CDS\ on clique-width promised in
Theorem \ref{thm:cw-main}. The dynamic programming table $D_i$ is a subset of
$B^w\times B^w\times B^w\times\{0,\ldots,n\}$. Informally,
$(\vec{a},\vec{u},\vec{d},c)\in D_i$ if there exists a set $S$ of
\emph{exactly} $c$ vertices and a partial mapping of vertices of $V\setminus S$
to $S$ such that

\begin{enumerate}

\item The total capacity of vertices of $S$ with label $l$ is roughly
$\vec{a}(l)$ 

\item The number of vertices with label $l$ in $V\setminus S$ mapped to $S$
(i.e.~dominated) is roughly $\vec{d}(l)$

\item The number of vertices of $V\setminus S$ mapped to vertices in $S$ with
label $l$ is roughly $\vec{u}(l)$.

\end{enumerate}

More simply, $\vec{a}$ is the total available capacity in a label set,
$\vec{u}$ is the capacity already used in our partial solution and $\vec{d}$ is
the number of vertices of each label set that we have already covered.

It is not hard to see how to fill this table for Initial, Rename and Union
nodes (only additions are needed). For Join nodes, the interesting point is
that we need to decide how many of the new edges are used. Similarly to the
case of \EDS, for each possible number $m$ we add $m$ to the \emph{used}
capacity $\vec{u}$ on one side and the number of \emph{dominated} vertices
$\vec{d}$ on the other. We drop solutions which are clearly (by more than a
factor $1+\eps$) invalid. Note that, because for each label set we only care
about its \emph{total} capacity, not the number of selected vertices, we can
afford to keep track of the total size of the dominating set exactly. This only
adds a polynomial factor to the algorithm's running time.

In the end we extract a solution from the root's table and argue about its
correctness as usual. One complication is that, since we do not know $\vec{a}$
and $\vec{d}$ exactly, the solution may be violating some capacities and it may
not be a complete dominating set.  We repair the solution by allowing some more
vertices to be undominated. If the total capacity of each label set was
violated by no more than $(1+\eps)$ this drops at most $\eps n$ vertices.
Similarly, if the solution was not a complete dominating set (and one exists),
at most $\eps n$ vertices are not dominated. We thus get the promised
bi-criteria guarantee.

Let us also discuss the case of treewidth (Theorem \ref{thm:tw-main}). Here the
dynamic program for \CDS\ is easier. We just need to remember which of the
vertices of each bag have been selected and how much of their capacity has
already been used. The result is a set that dominates the whole graph, has size
at most $\OPT$, but may violate some capacities by $(1+\eps)$. Observe that
this is stronger than the bi-criteria approximation we get for clique-width
since, again, we could drop the coverage for some vertices to fix the
capacities (but in general, it's not clear if we can trade in the other
direction).

It's worthy of note that for \CDS\ we can probably not hope to obtain anything
better than a bi-criteria approximation, such as the one we gave here. The
reason is that in \cite{DomLSV08} it is shown that \CDS\ is W-hard even if
parameterized by \emph{both} the treewidth and $\OPT$. Thus, if we could obtain
an FPT approximation scheme for $\OPT$ without violating constraints, we would
be able to set $\eps$ to an appropriate small value and obtain an FPT algorithm
for a W-hard problem.

Finally, let us mention the \CVC\ algorithm promised in Theorem
\ref{thm:tw-main}.   This can be obtained by reducing \CVC\ to a
vertex-weighted version of \CDS. First, observe that if we associate with each
vertex of an instance of \CDS\ a non-negative integer cost and ask for a
solution of minimum total cost the algorithm we gave still works with minimal
modification. Consider now a \CVC\ instance. First, sub-divide each edge (this
does not increase the treewidth $w$) and set the capacities of new vertices to
$0$.  Then, add a new vertex $u$ and connect it to all of the original vertices
of the instance (this increases $w$ by 1). Set the capacity of $u$ to be equal
to its degree.  Finally, set the cost of each original vertex to $1$, the cost
of $u$ to $0$ and the cost of vertices constructed in the sub-divisions to $n$.
If we view this as an instance of weighted \CDS\ it is straightforward that a
solution of cost $s<n$ exists if and only if the original \CVC\ instance had a
solution of size $s$. 

It would be interesting to see if an algorithm for \CVC\ on clique-width can be
given. At the moment this seems more challenging than for \CDS, because it is
not sufficient to remember the total available capacity of a whole label set.

\subsection{Bounded Degree Deletion}

A dynamic program for clique-width can be given as follows: each table is a
subset of $B^w\times B^w\times \{0,\ldots,n\}$. Informally, a partial solution
is a set of ``active'' vertices which will \emph{not} be deleted. We remember
how many of these we have in each label set. We also remember what is the
maximum degree of any active vertex in each label set. Finally, we remember the
total number of vertices we have deleted. The only interesting case is Join
nodes, where the new maximum degree is calculated by adding to the previous
maximum degree the number of active vertices in the other label set.

A dynamic program for treewidth can be formulated by keeping track of the
active vertices inside a bag. For each one of these we remember (approximately)
the number of active neighbors it has in the set of vertices that appear in
bags lower in the decomposition. The running time for both clique-width and
treewidth is $(\log n/\eps)^{O(w)}$.

Let us remark that, even though we only present a bi-criteria approximation
algorithm here, to the best of our knowledge there is no complexity barrier
ruling out the existence of an approximation scheme for the natural objective
of this problem (the size of the deletion set). 

\subsection{Graph Balancing}

We only deal with this problem for treewidth. We assume that edge weights are
written in unary (this is the interesting case that is usually studied to avoid
unnecessary complications).  For a bag of the tree decomposition, the signature
of a partial solution is the weighted out-degree each vertex of the bag has
already accumulated due to edges whose other endpoint appears lower in the
decomposition. We keep track of this information and the maximum seen so far,
making the table be a subset of $B^{w+1}$. When a vertex is forgotten (that is,
we move to the first bag that does not contain it) we go through every
orientation of the edges connecting it to the rest of the bag and update the
out-degrees of other vertices accordingly. The running time is $(\log
n/\eps)^{O(w)}$ and we can extract a solution with out-degree at most
$(1+\eps)\OPT$.

Let us remark that for the case of clique-width, it's not possible to achieve a
similar result. It is known that even for graphs with edge weights 1 and 2 it
is NP-hard to tell if $\OPT$ is 2 or 3 \cite{EbenlendrKS08}.  Take an arbitrary
such instance, multiply all edge weights with $n^2$ and then replace every
non-edge with an edge of weight 1. It is now NP-hard to tell if $\OPT$ is at
most $2n^2+n$ or at least $3n^2$, so a better than $3/2$ approximation is
NP-hard even for cliques.

\section{Conclusions}

We have presented a generic technique which can be applied with minor
modifications to a number of hard problems. The question now becomes to which
other problems we can apply this technique. The most prominent next target is
\textsc{Hamiltonicity} parameterized by clique-width.  This is another example
that separates clique-width from treewidth in term of parameterized
tractability, and its natural dynamic program uses integers. Unfortunately, the
natural program also uses subtractions, so it does not immediately yield to our
methods. In several of the problems of this paper we managed to work around
such obstacles, rewriting the natural program to only use additions. Is this
possible here, or are completely new ideas needed?

Let us also mention that in this paper we have focused on establishing that our
algorithms run in $(\log n/\eps)^{O(w)}$, but we have not been paying much
attention to the constant hidden in the exponent.  Another good direction would
therefore be to improve the analysis of Approximate Addition Trees in order to
obtain better running times for our schemes. Perhaps combining this with the
idea of (partially) balancing the given decomposition could also help speed up
the overall algorithms.

\subsubsection*{Acknowledgment} I am grateful to an anonumous reviewer for
pointing out that treewidth balancing theorems can be used to obtain some of
the results of this paper in a simpler way.

\bibliographystyle{plain}
\bibliography{fpt-as}

\end{document}